\documentclass[submission,copyright,creativecommons]{eptcs}

\usepackage{iftex}

\ifpdf
  \usepackage{underscore}         
  \usepackage[T1]{fontenc}        
\else
  \usepackage{breakurl}           
\fi

\usepackage[nocompress]{cite}

\usepackage{etoolbox}
\usepackage{amssymb}
\usepackage{mathtools}
\usepackage{mathrsfs}    
\usepackage{caption}
\usepackage[normalem]{ulem}

\usepackage{amsthm}

\usepackage{mathpartir}
\usepackage{stmaryrd}
\usepackage{proof}

\usepackage[dvipsnames]{xcolor}
\usepackage{xspace}                

\DeclareMathAlphabet{\mathcal}{OMS}{cmsy}{m}{n}
\newcommand{\mc}{\mathcal}

\newtheorem{definition}{Definition}
\newtheorem{theorem}{Theorem}
\newtheorem{lemma}{Lemma}

\newtheorem{corollary}{Corollary}

\newcounter{itemnum}

\newcommand{\case}[1]{\vskip12pt\noindent{#1:}}
\newcommand{\refcase}[2][]{\vskip12pt\noindent{\hyperref[#1]{#2:}}}

\newcommand{\st}{%
  \nonscript\;
  \ifnum\currentgrouptype=16
    \middle\vert
  \else
    \mid
  \fi
  \nonscript\;}

\newcommand{\ie}{{\it i.e.}}
\newcommand{\eg}{{\it e.g.}}
\newcommand{\seq}[3]{\ifstrempty{#1}{\cdot}{#1}; \ifstrempty{#2}{\cdot}{#2} \vdash #3}
\newcommand{\deriv}[2]{{\begin{array}{c}#1\\#2\end{array}}}
\newcommand{\gseq}[2]{\ifstrempty{#1}{\cdot}{#1}\vdash #2}

\newcommand{\existsx}[2]{\exists #1.\,#2}
\newcommand{\quantx}[3]{Q #2:#1.#3}
\newcommand{\typedforallx}[3]{\forall #2:#1.#3}
\newcommand{\typedexistsx}[3]{\exists #2:#1.#3}
\newcommand{\imp}{\supset}

\newcommand{\ind}{\stackrel{\mathclap{\scalebox{0.66}{\mbox{$\mu$}}}}{=}}
\newcommand{\defn}{\stackrel{\mathclap{\scalebox{0.5}{\mbox{$\Delta$}}}}{=}}

\newcommand{\dfn}{\mathrm{defn}}
\newcommand{\fpdfn}[3]{#1 \defn_{#2} #3}
\newcommand{\inddfn}[3]{#1 \ind_{#2} #3}

\newcommand{\IR}{$\mu\mc R$\xspace}
\newcommand{\IL}{$\mu\mc L$\xspace}
\newcommand{\muR}{\IR}
\newcommand{\muL}{\IL}

\newcommand{\DR}{$\Delta\mc R$\xspace}
\newcommand{\DL}{$\Delta\mc L$\xspace}
\newcommand{\axiom}{\textsc{Ax}\xspace}

\newcommand{\multc}{$\textit{mc}$\xspace}

\newcommand{\impR}{$\imp\!\!\mc R$\xspace}

\newcommand{\allR}{$\forall\mc R$\xspace}
\newcommand{\allL}{$\forall\mc L$\xspace}
\newcommand{\exR}{$\exists\mc R$\xspace}
\newcommand{\exL}{$\exists\mc L$\xspace}

\newcommand{\multcut}{mc}

\newcommand{\G}{$\mc G$\xspace}
\newcommand{\LD}{$\text{LD}$\xspace}

\newcommand{\LDI}{$\text{LD}^{\mu}$\xspace}
\newcommand{\LDIinf}{$\text{LD}^{\mu}_\infty$\xspace}

\newcommand{\NJI}{$\mu\text{NJ}$\xspace}

\newcommand{\lvl}{\mathrm{lvl}}
\newcommand{\ground}{\mathrm{ground}}

\newcommand{\range}{\mathrm{range}}

\newcommand{\lambdax}[2]{\lambda #1.\,#2}
\newcommand{\app}{\ }

\newcommand{\mkconst}[1]{\mbox{\sl \color{purple}{#1}}}
\newcommand{\lamtm}{\mkconst{lam}}
\newcommand{\apptm}{\mkconst{app}}
\newcommand{\suc}{\mkconst{s}}
\newcommand{\zero}{\mkconst{z}}
\newcommand{\z}{\zero}
\newcommand{\nil}{\mkconst{nil}}
\newcommand{\cons}{\mkconst{cons}}
\newcommand{\unit}{\mkconst{unit}}
\newcommand{\arr}{\mkconst{arrow}}
\renewcommand{\star}{\mkconst{star}}

\newcommand{\mktype}[1]{{\mbox{\sl {#1}}}}
\renewcommand{\o}{\mktype{\bf o}}
\newcommand{\tm}{\mktype{tm}}
\newcommand{\ty}{\mktype{ty}}
\newcommand{\nat}{\mktype{nat}}

\newcommand{\mkpred}[1]{\mbox{\sl \color{RoyalBlue}{#1}}}
\newcommand{\ev}{\mkpred{ev}}
\newcommand{\odd}{\mkpred{odd}}
\newcommand{\eq}{\mkpred{eq}}
\newcommand{\append}{\mkpred{append}}
\newcommand{\red}{\mkpred{red}}
\newcommand{\sn}{\mkpred{sn}}
\newcommand{\step}{\mkpred{step}}
\newcommand{\type}{\mkpred{type}}
\newcommand{\p}{\mkpred{p}}

\title{Ground Stratification for a \\ Logic of Definitions with Induction}
\author{Nathan Guermond\qquad\qquad Gopalan Nadathur
  \email{guerm001@umn.edu\qquad\qquad\qquad ngopalan@umn.edu}
  \institute{University of Minnesota Twin-Cities\\
  Minnesota, USA}
  }

\newcommand{\titlerunning}{Ground Stratification for a Logic of Definitions with Induction}
\newcommand{\authorrunning}{N. Guermond and G. Nadathur}

\newcommand{\event}{LFMTP 2025}

\hypersetup{
  bookmarksnumbered,
  pdftitle    = {\titlerunning},
  pdfauthor   = {\authorrunning},
  pdfsubject  = {LFMTP},               
}
\begin{document}
\maketitle

\begin{abstract}
The logic underlying the Abella proof assistant includes mechanisms
for interpreting atomic predicates through fixed point definitions
that can additionally be treated inductively or co-inductively.
However, the original formulation of the logic includes a strict
stratification condition on definitions that is too restrictive for
some applications such as those that use a logical relations
based approach to semantic equivalence.
Tiu has shown how this restriction can be eased by utilizing a weaker
notion referred to as ground stratification.
Tiu's results were limited to a version of the logic that does not
treat inductive definitions.
We show here that they can be extended to cover such definitions. 
While our results are obtained by using techniques that have been
previously deployed in related ways in this context, their use is
sensitive to the particular way in which we generalize 
the logic.
In particular, although ground stratification may be used with 
arbitrary fixed-point definitions, we show that weakening
stratification to this form for inductive definitions leads to
inconsistency.
The particular generalization we describe accords well with the way
logical relations are used in practice.
Our results are also a intermediate step to building a more flexible
form for definitions into the full logic underlying Abella, which
additionally includes co-induction, generic quantification, and a
mechanism referred to as nominal abstraction for analyzing occurrences
of objects in terms that are governed by generic quantifiers.
\end{abstract}

\section{Introduction}\label{sec:intro}

This paper concerns a family of first-order predicate logics that
originate from the work of McDowell and Miller~\cite{MDM00} and that
have culminated in the logic \G that underlies the proof assistant
Abella~\cite{GMN11}.
These logics endow a conventional predicate logic with the capability
of treating predicate constants as defined symbols, to be interpreted
via formulas that have been associated with them by a definition.
More specifically, the definitions that are associated with predicate
constants are given a fixed-point reading that can be further refined
to correspond to the least or greatest fixed-point, thereby adding the
capability of inductive and co-inductive reasoning to the logic.
These logics have been especially useful in encoding and reasoning
about object systems that are described in a rule-based and relational
fashion: the rules in the object system description can be translated
into definitions of predicate constants that represent the relevant
relations and the treatment of definitions in the logic provides a
transparent means for capturing the informal style of reasoning
associated with rule-based specifications.

An important aspect of the logics of interest is that the forms of
definitions must be sufficiently constrained to ensure consistency.
The definition of a predicate constant can include a use of that
constant itself, thereby supporting recursive specifications.
However, a definition must not be circular in that it assumes its own
existence in its construction.
In logics of definitions, this requirement translates into
restrictions on the negative occurrences of the predicate constant in
the formula defining it.
The condition imposed in the logic described in \cite{MDM00}, which
has carried over to the logic \G, is that the predicate constant must
not appear to the left of a top-level implication symbol in that
formula. 
Concretely, this condition takes the form of a \emph{stratification}
requirement: there must be an ordering on predicate
constants that determines which of them have been already defined and
can therefore be used negatively in the definition of another
predicate constant. 
Unfortunately, the stratification restriction based on predicate names
is too strong for some applications.
A prime example of this is that of logical relations that are often
used in reasoning about programming languages properties.
The definitions of these relations are typically indexed by types 
and assume their own definition, albeit for structurally smaller
types. 

This work is part of the effort to allow for more permissive definitions
so as to support the described applications.
In past work~\cite{Tiu12}, Tiu has developed the notion of
\emph{ground stratification} that effectively allows for the building
in of the arguments of an atomic formula into the stratification ordering.
The logic considered by Tiu limits definitions to a generic
fixed-point variety with inductive reasoning being realized through
natural number induction.
The results in this paper are to be viewed in the context of Tiu's
work and have a twofold character.
First, we show that stronger conditions must be satisfied by
definitions in the case of predicate constants that are to be treated 
inductively.
In particular, we show that easing the restriction to ground
stratification for these constants can lead to an inconsistent logic.
Second, we show that if a stronger form of stratification for such
constants is coupled with ground quantification for predicate
constants that are interpreted via generic fixed-points, then the
logic is consistent. 
Proofs of consistency for logics typically proceed by showing a
property called cut-elimination for them.
However, such a proof is elusive for our logic.
Instead, we reduce consistency for it to consistency for a ground
version, for which we show a cut-elimination result.
In providing such a proof, we follow the lead of \cite{Tiu12}, using 
ideas from \cite{Tiu04} in the additional treatment of induction.

The rest of the paper is organized as follows.
In the next section we present the logic of study, describing in its
context the particular mix of stratification conditions that are
needed for consistency.
Section~\ref{sec:strat-and-consistency} exposes the need for the
restrictions imposed on the allowable definitions.
Section~\ref{sec:examples} presents some examples that show, amongst
other things, that the logic described is capable of encoding a
reasoning style that is based on using logical relations.
Section~\ref{sec:consistency} sketches the proof of consistency for
the logic.\footnote{The complete development of these ideas can be
found at the URL
\href{https://z.umn.edu/strat-proofs}{https://z.umn.edu/strat-proofs}.}
We conclude the paper in Section~\ref{sec:conclusion} with a discussion of
related and future work.

\section{A Logic with Definitions and Induction}\label{sec:logic}

The logic that we consider in this paper, which we call \LDI, has 
an intuitionistic, first-order version of Church's Simple Theory of
Types~\cite{Ch40} as its core.
It extends this logic by allowing atomic predicates to be
treated as defined symbols, to be interpreted via a collection of
clauses that are given a fixed-point interpretation, which can be
further refined to correspond to the least fixed-point; the
restriction to the least-fixed point amounts to giving the clauses
defining the predicate an inductive reading. 
Qualitatively, \LDI is an extension of Tiu's \LD with inductive
definitions.
It can also be viewed as a fragment of the logic \G~\cite{GMN11},
which underlies the Abella proof assistant~\cite{Baelde14jfr}, that
has been enriched with a more permissive form for definitions.
In the subsections below, we outline this logic, focusing mainly on
the notions of definitions and, more specifically, on how they extend
what is permitted in \G. 

\subsection{Syntax}

The terms of \LDI are based on the expressions in
Church's Simple Theory of Types, in which we assume a finite set of base
types $\iota_1,\ldots,\iota_n$, function types $\alpha\to\beta$, and a
distinguished type $\o$ of propositions.
We assume we are given a \emph{signature} $\Sigma$, which is a set of
type annotated constants of the form $\kappa : \tau$.
Terms in \LDI are then the well-typed terms in the simply typed lambda
calculus that are constructed using the constants in the signature
$\Sigma$ and variables from a \emph{variable context} $\mc 
X$, which is a finite set of distinct type annotated variables $x :
\tau$.
Since a term $t$ only exists with respect to a signature
$\Sigma$ and a variable context $\mc X$, and since $\Sigma$ is fixed,
we will say that $t$ is well formed with respect to $\mc X$, or that
$t$ \emph{lies over} $\mc X$.
In particular this means that $\mc X$ must contain all free variables
in $t$, but may contain more.
A term is \emph{ground} if it lies over the empty variable context
$\emptyset$.
We denote the set of ground terms of type $\alpha$ by
$\ground(\alpha)$.
We observe that the type of $t$ is uniquely determined by
$\Sigma$ and $\mc X$, and thus we denote $\mc X \vdash_\Sigma t :
\tau$ whenever $t$ lies over $\mc X$.
Finally, we note that terms are considered to be equal modulo
$\alpha$-, $\beta$-, and $\eta$-conversion and we shall represent them
by their normal forms modulo these rules, which are known to exist.

We call a type a \emph{first-order} type if it does not contain $\o$,
and a \emph{predicate} type if it is $\o$ or
of the form $\tau\to \omega$ for a first order type $\tau$ and a
predicate type $\omega$.
If $p:\omega$ belongs to $\Sigma$ for a predicate type $\omega$, then
$p$ is said to be a predicate constant or symbol.
If $\omega$ is $\o$, then $p$ may also be referred to as a
propositional symbol.
A \emph{formula} is a term of type $\o$.
We assume that $\Sigma$ contains the \emph{logical constants} $\bot$
and $\top$ of type $\o$, $\wedge$, $\vee$, and $\imp$ of type $\o\to
\o\to\o$ and, for every first-order type $\alpha$, $\forall_\alpha$
and $\exists_\alpha$ of type $(\alpha\to \o)\to \o$.
\emph{Atomic} formulas are of the
form $p\ \vec t$ for some predicate symbol $p : \omega$ in
$\Sigma$ where in general, $u\ \vec t$ is an abbreviation for the
application $(\ldots(u\ t_1)\ \ldots\ t_n)$. We abbreviate
$Q_\alpha(\lambda x.C)$, where $Q_\alpha$ is $\forall_\alpha$ or
$\exists_\alpha$, by $\quantx{\alpha}{x}{C}$.
We will also follow the usual convention of writing $\wedge$, $\vee$,
and $\imp$ in infix form in formulas in which they appear as the
top-level logical symbol.

Given variable contexts $\mc X$ and $\mc Y$, a \emph{substitution}
$\theta = [t_1/x_1,\ldots,t_n/x_n]$ of type $\mc Y\to\mc X$ is an
assignment of terms $t_1 : \tau_1,\ldots,t_n:\tau_n$ lying over $\mc
Y$ to each variable $x_1:\tau_1,\ldots,x_n:\tau_n = \mc X$.
Such a substitution is said to be \emph{for} $\mc X$, and its
\emph{range}, denoted by $\range(\theta)$, is  the smallest variable
context containing all the free variables in $t_1,\ldots,t_n$. 
Moreover, it is said to be an \emph{$\mc{X}$-grounding substitution} if $\mc
Y$ is $\emptyset$.
If $t$ is a term lying over $\mc X$, the \emph{application of $\theta$
to $t$}, written as $t\theta$, is the term $(\lambda x_1\ldots \lambda
x_n. t)\ \vec t$; this is evidently a term that lies over $\mc Y$.
Given a variable context $\mc X$, we write $[t/x]_{\mc X}$ to denote the
simultaneous substitution $[x_1/x_1,\ldots,x_n/x_n,t/x] : \mc X\to \mc
X,x$, and we denote the trivial substitution by $\epsilon_{\mc X}:\mc
X\to\mc X$.

We assume a sequent style formulation of derivability for \LDI.
In this context, a \emph{sequent} is a judgment of the form
$\seq{\mc X}{\Gamma}{C}$ for a finite multiset of formulas $\Gamma$
and a formula $C$ all of which lie over a finite variable context $\mc
X$. 
We will assume all the types in $\mc X$ are first order.
We will omit the type annotations in $\mc X$ when their specific
identity is orthogonal to the discussion. 
The formulas in $\Gamma$ are called the \emph{context} or
\emph{assumption set} of the sequent, $C$ is called its
\emph{conclusion} and the variables in $\mc X$ constitute its
\emph{eigenvariables}. 
At an intuitive level, the judgment represented by such a sequent is
valid if, for every $\mc{X}$-grounding substitution $\theta$, the
validity of all the formulas in $\Gamma\theta$ implies the validity of
$C\theta$.

\subsection{Logical Rules}

The main content of a sequent style presentation of a logic are the
rules for deriving sequents.
The rules for \LDI are of three varieties: those that explicate the
meaning of the logical symbols, those that deal with the structural
aspects of sequents,  and those that build in the treatment
of definitions.
We discuss rules of the first two kinds here, leaving the elaboration
of definitions to the next two subsections.

\begin{figure}[!ht]
\begin{mathpar}
  {\inferrule*[right=\allL]
    {\seq{{\mc X}}{\Gamma,C[t/x]_{\mc X}}{D}\\
      \mc X \vdash_\Sigma t : \tau}
    {\seq{{\mc X}}{\Gamma,\typedforallx{\tau}{x}{C}}{D}}}\and
  {\inferrule*[right=\allR]
    {\seq{{\mc X, x}}{\Gamma}{C}}
    {\seq{{\mc X}}{\Gamma}{\typedforallx{\tau}{x}{C}}}}\\
  {\inferrule*[right=\exL]
    {\seq{{\mc X,x }}{\Gamma,C}{D}}
    {\seq{{\mc X}}{\Gamma,\typedexistsx{\tau}{x}{C}}{ D}}}\and
  {\inferrule*[right=\exR]
    {\seq{{\mc X}}{\Gamma}{C[t/x]_{\mc X}}\\
      \mc X\vdash_\Sigma t : \tau}
    {\seq{{\mc X}}{\Gamma}{\typedexistsx{\tau}{x}{C}}}}
\end{mathpar}
\caption{Logical rules for quantifiers}
\label{fig:quantifiers}
\end{figure}

\begin{figure}[!ht]
  \begin{mathpar}
    {\inferrule*[right=\multc]
      {\seq{\mc X}{\Delta_1}{ A_1}
        \qquad\ldots\qquad
        \seq{\mc X}{\Delta_n}{ A_n}
        \and
        \seq{\mc X}{\Gamma,A_1,\ldots,A_n}{ C}}
      {\seq{\mc X}{\Gamma,\Delta_1,\ldots,\Delta_n}{ C}}}\and
    {\inferrule*[right=\axiom]{ }{\seq{\mc X}{A}{A}}}\quad\raisebox{.75em}{$A$\text{ atomic}}
  \end{mathpar}
  \caption{The multicut and axiom rules}
  \label{fig:multicut-axiom}
\end{figure}

The rules corresponding to the logical symbols are identical in
content to the ones to be found in usual first-order logics.
We assume familiarity with the ones that build in the meanings of the
logical connectives.
The quantifier rules are shown in Figure~\ref{fig:quantifiers}.
Note that we write $\Gamma, F$ in these rules to denote a multiset
that comprises $F$ and the formulas in $\Gamma$.
Note also that, in the rules in Figure~\ref{fig:quantifiers}, we
assume that the variable $x$ bound by the quantifier to be distinct
from all the variable in $\mc X$, a requirement that can always be
established by renaming the bound variable.
The structural rules include the usual complement:
contraction, which builds in the treatment of multisets as sets,
weakening, which allows us to add an extraneous assumption,
an axiom rule which allows us to match an atomic conclusion with an
assumption set comprising just that formula, and the cut rule that
underlies the use of lemmas.
We use a particular version of the cut rule called \emph{multicut}
that is better suited to cut-elimination and consistency arguments.
The axiom rule and the multicut rule are shown in
Figure~\ref{fig:multicut-axiom}. 

If an inference rule is a right introduction rule, all of its premises are
called \emph{major} premises. If it is a left introduction rule or a
multicut, then only those premises with the same consequent as its conclusion
are called major premises. All other premises are called \emph{minor} premises.

\subsection{Treating Predicate Constants as Defined Symbols}
\label{sec:definitions}

\LDI deviates from a vanilla first-order logic in that it allows
atomic predicates to be further analyzed through a \emph{definition}
$\mc D$ that parameterizes the logic.
In this context, a definition is a set of clauses of the form
$\fpdfn{p\ \vec t}{{{\mc X}}}{B}$ for which there exists a predicate
constant $p : \omega$ in the signature $\Sigma$ and a formula $B$, with
$\vec t$ and $B$ both lying over $\mc X$.
For a given clause $\fpdfn{H}{{\mc X}}{B}$, we say that $H$ is the
\emph{head} of the clause, and $B$ the \emph{body}.
Similarly to \cite{Tiu12}, we require every variable in $\mc
X$ to appear in $H$, and $H$ to lie in the higher-order pattern
fragment that is described, \eg, in \cite{MN12}.

Definitions provide a natural means for encoding rule-based relational
specifications in \LDI.
For example consider the specification of the \append\ relation
over lists that are constructed using the constants \nil, that
represents the empty list, and \cons, that represents the construction
of a new list by adding an element to an already existing
list.
A typical specification of this relation comprises the following
rules:
  \[
    \inferrule*
        {\qquad}
        {\append\app \nil\app K\app K} \qquad\qquad
    \inferrule*
        {\append\app L\app K\app M}
        {\append\app (\cons\app X\app L)\app K\app
          (\cons\app X\app M)}
  \]

Such a specification can be rendered into a definition through the
following clauses in \LDI:
\begin{equation}
  \label{eq:append-clausal}
  \begin{aligned}
  \append\ \nil\ K\ K &\defn_{K} \top\\
  \append\ (\cons\ X\ L)\ K\ (\cons\ X\ M)&\defn_{X,L,K,M} \append\ L\ K\ M
  \end{aligned}
\end{equation}

Of course, such a rendition is useful only if it is supplemented with 
a means for reflecting the natural style of reasoning associated
with rule-based specifications into \LDI.
There are, in general, two forms in which the rules might be used in
informal reasoning.
First, they may be used to construct derivations for particular
relations; for example, in this particular instance, they may be used
to show that the relation $\append\app (\cons\app a\app \nil)\app
(\cons\app b\app \nil)\app (\cons \app a\app (\cons\app b\app \nil))$
holds for particular constants $a$ and $b$.
Second, when we are given as an assumption that a particular relation
holds, they can figure in a case analysis style of reasoning.
Thus, in this particular instance, they can be used to show that the
formula $(\append\app (\cons\app a\app \nil)\app \nil\app \nil)\imp
\bot$ must hold because there cannot be a derivation for its
antecedent.

\begin{figure}[htb!]
\begin{mathpar}  
  {\inferrule*[right=\DL]
    {\{\seq{\range(\theta)}{\Gamma\theta,B'}{C\theta}\st
      \dfn(\fpdfn{H}{{\mc X}}{B},A,\theta,B'), \fpdfn{H}{{\mc X}}{B}
      \in {\mc D} \}}
    {\seq{{\mc Y}}{\Gamma, A}{C}}}\\
  {\inferrule*[right=\DR]
    {\seq{{\mc Y}}{\Gamma}{B'}}
    {\seq{{\mc Y}}{\Gamma}{A}}\quad \raisebox{1em}{$\dfn(\fpdfn{H}{\mc
        X}{B},A,\epsilon_{\mc Y},B')\ \mbox{\rm for}\ \fpdfn{H}{\mc
        X}{B}\in {\mc D}$}}
\end{mathpar}
\caption{Definition rules for a predicate $p$, provided $A = p\ \vec t$}
\label{fig:defn-rules}
\end{figure}

These two forms of reasoning are built into \LDI by rules for
introducing atomic predicates into the left and right of a sequent
based on the definition $\mc D$ that parameterizes the logic.
These rules are shown in Figure~\ref{fig:defn-rules}.
In these rules, we use the notation $\dfn(\fpdfn{H}{{\mc X}}{B},
A,\theta,B')$ to signify that there exists a substitution
$\rho:\mc Z\to \mc X$, for some $\mc Z$, such that $H\rho = A\theta$
and $B' = B\rho$
Conceptually, \DR corresponds to unfolding a definition for a
particular instance $A$ of the head $H$, whereas \DL corresponds to
case analysis on all possible instances of $A$ matching with the head
$H$.
The reader may confirm that these rules can actually be used to
establish the two formulas considered above.

Not all definitions are permissible in \LDI.
To explain what definitions are allowed, we must describe a process
for assigning an ordinal to each ground
formula $F$ that is called its \emph{level} and is designated by
$\lvl(F)$. 
This process assumes an assignment for each ground atomic formula and
extends it to arbitrary formulas based on the following rules:
\begin{align*}
  &\lvl(\bot) := \lvl(\top) := 0 \qquad&\\
  &\lvl(A\wedge B) := \lvl(A\vee B) := \max(\lvl(A),\lvl(B))&\\
  &\lvl(A\imp B) := \max(\lvl(A) + 1,\lvl(B))&\\
  &\lvl(\typedforallx{\alpha}{x}{C}) := \lvl(\typedexistsx{\alpha}{x}{C}) := \sup\{\lvl(C[t/x]_{\emptyset})\st t\in\ground(\alpha)\}&
\end{align*}
We then say that a definition $\mc D$ is \emph{ground stratified} if
there exists a level assignment to ground atomic formulas such that 
for every clause $H\defn_{\mc X} B\in \mc D$, and for every $\mc
X$-grounding substitution $\rho$, it
is the case that $\lvl(H\rho)\geq \lvl(B\rho)$.
The requirement of definitions in \LDI, then, is that they must be
ground stratified.

There is a stronger version of stratification that is also of interest
and that we identify as \emph{strict stratification}.
In this version, we require the definition to be ground stratified
under a level assignment to ground atomic formulas that depends
only on their predicate head.
In other words it must be the case that $\lvl(p\app \vec{t_1}) =
\lvl(p\app \vec{t_2})$ for any sequence of ground terms $\vec{t_1}$ and
$\vec{t_2}$.
It is actually this less permissive version of stratification that
governs definitions in the logic \G. 

\subsection{Inductive Definitions and Fixed-Point Operators}
\label{sec:induction}

Definitions as we have considered them thus far may contain only one
kind of clause: those of the form $\fpdfn{H}{\mc X}{B}$.
We now introduce the possibility of including a different kind of
clause in definitions, ones that are written as
$\inddfn{H}{\mc X}{B}$.
We refer to these as \emph{inductive clauses} in contrast to the
previously described ones that we distinguish as \emph{fixed-point
clauses}.
The structural restrictions on inductive clauses parallel those on
fixed-point clauses: all the variables in ${\mc X}$ must appear in
$H$ and $H$ must lie in the higher-order pattern fragment.
The predicate symbols in $\Sigma$ are categorized as inductive or
fixed-point predicates.
The definition that parameterizes the logic can be a mixture of the
two kinds of clauses with the proviso that the clauses defining an
inductive predicate must all be of the inductive variety and,
similarly, those for the fixed-point predicates must be of the
fixed-point variety.

Definitions in this mixed form must once again satisfy a stratification
condition.
The condition is, in the first instance, similar to that described
earlier: the definition must be ground stratified under a level
assignment to ground atomic formulas.
This initial assignment is permitted to be arbitrary for ground atomic
formulas that have a fixed-point predicate as the top-level predicate
symbol.
However, the requirement is stricter when the top-level predicate is
an inductive one.
In this case, the level assignment must be independent of the arguments of
the predicate, \ie, the conditions for strict stratification must be
satisfied in these cases.

The conceptual difference between fixed-point and inductive clauses is
that the latter are intended to identify a \emph{least} fixed-point
which thereby entails stronger properties for the predicate.
At the proof-theoretic level, this more refined view of the predicate
definition is realized by a special introduction rule for assumption
formulas that have an inductively defined predicate as their top-level
predicate symbol.
To present this rule we must first render a multi-clause definition of
an inductive predicate $p$ into a single clause that has the form
$\inddfn{p\app \vec{x}}{\vec{x}}{B\app p\app \vec{x}}$,
where $\vec x$ is a sequence of distinct variables and $B$ is a closed
term not containing $p$; $B$ is referred to as a \emph{fixed-point
operator} in such a definition. 
This translation makes use of a special predicate $\eq$ that is
assumed to be defined by the sole clause 
$\fpdfn{\eq\app X\app X}{X}{\top}$.
More specifically, if the clauses for $p$ are the following
\[\inddfn{p\ \vec t_1}{\mc X_1}{B_1} \quad\ldots \quad
  \inddfn{p\ \vec t_n}{\mc X_n}{B_n}\]
then the fixed-point operator in the single clause
definition of $p$ would be
\begin{align*}
  \lambda p. \lambda \vec{x}. &
            (\exists \mc X_1.(\eq\ x_1\ t^1_1)\wedge\ldots\wedge(\eq\ x_k\ t^k_1) \wedge
                      B_1)\vee\ldots\vee\\ 
           &(\exists \mc X_n.(\eq\ x_1\ t^1_n)\wedge\ldots\wedge(\eq\ x_k\ t^k_n)\wedge B_n)
\end{align*}
where $\vec t_i = t_i^1,\ldots,t_i^k$ for each $i$, with $t_i^j$ lying
over $\mc X_i$.

To provide a concrete example of this translation, we may consider the
clauses shown in the display labelled \ref{eq:append-clausal} that
define the \append\ predicate.
Those clauses were originally shown to be of the fixed-point variety,
but we will now assume that their annotation and designation has been
changed to that of inductive clauses.
Letting $B$ be the fixed-point operator
\begin{align*}
    \lambda p.\lambda \ell .\lambda k.\lambda m. & ((\eq\ \ell\ \nil) \wedge (\eq\ k\ m))\vee\\
                     &(\exists x,\ell',m'. (\eq\ \ell\ (\cons\ x\ \ell'))\wedge
    (\eq\ m\ (\cons\ x\ m'))\wedge (p\ \ell'\ k\ m')),
\end{align*}
this definition would be transformed into $\inddfn{\append\app
  L\app K\app M}{L,K,M}{((B\app \append)\app L\app K\app M)}$.

\begin{figure}[ht]
\begin{mathpar}
  {\inferrule*[right=\IL]
    {\vec x; B\ S\ \vec x\vdash S\ \vec x\qquad \mc X; \Gamma,S\ \vec t\vdash C}
    {\mc X; \Gamma, p\ \vec t\vdash C}}\and
  {\inferrule*[right=\IR]
    {\mc X;\Gamma\vdash B\ p\ \vec t}
    {\mc X;\Gamma\vdash p\ \vec t}}
\end{mathpar}
\caption{Rules for introducing $p\app\vec{t}$ after converting the
  clauses for $p$ into the form $\inddfn{p\app\vec{x}}{\vec{x}}{B\ p\ \vec{x}}$} 
\label{fig:ind-rules}
\end{figure}

Assuming the translation that we have just described, the rules for
introducing atomic formulas that have an inductively defined predicate
as their top-level predicate symbol are shown in Figure~\ref{fig:ind-rules}.
The symbol $S$ that appears in the left premise sequent in the \IL
rule is required to be instantiated with a closed term of the same
type as the predicate $p$.
The particular term that is chosen for $S$ in a use of this rule is
referred to as \emph{inductive invariant.}

It is easily seen that the \IR rule, \ie, the rule for introducing an
inductively defined atomic formula in the conclusion of a sequent, is
no different from the rule for doing so when the formula is defined as
simple fixed-point.
However, the \IL rule adds deductive power to the calculus beyond what
is available from just fixed-point definitions.
To substantiate this observation, consider showing the following
statement, which establishes that $\append$ is a functional relation
$$\forall \ell,k,m,m'. (\append\ \ell\ k\ m)\wedge (\append\ \ell\ k\ m')\imp
(\eq\ m\ m').$$
This property cannot be proved if \append\ is
defined via clauses that are given just a fixed-point interpretation.
However, if we view the clauses for \append\ as inductive ones
instead, then the property can be proved via the \IL rule, using 
$$\lambda \ell.\lambda k.\lambda m. \forall m'.(\append\ \ell\ k\ m)\wedge (\append\ \ell\ k\ m')\imp (\eq\ m\ m')$$
as the inductive invariant.

\section{Stratification and
  Consistency}\label{sec:strat-and-consistency} 

In Section~\ref{sec:logic}, we imposed some conditions on the forms of
fixed-point and inductive definitions.
In particular, we required the first to be ground stratified and the
second to be stratified based on just the names of predicates.
In this section, we discuss the purpose of these restrictions.

The need for some kind of stratification for the coherence of
definitions should not be difficult to appreciate.
As an example of a definition that might be problematic, consider one
that has $\fpdfn{\p}{}{\p \imp \bot}$ as the sole clause for a
propositional symbol $\p$.
A clause of this kind has the intuitive content of assuming a
definition of $\p$ in the course of constructing one for $\p$.
To understand this, consider how we might proceed to prove $\p$.
We might unfold this by virtue of the definition into trying to prove
$\p \imp \bot$, which could be done by proving the sequent
$\seq{}{\p}{\bot}$.
Since $\p$ appears as a hypothesis in this sequent, we
are effectively assuming that $\p$ is already a defined propositional
symbol.

The intuitive observation above can be given actual substance by
showing that the sequent $\seq{}{}{\bot}$ can be derived if we
allow $\fpdfn{\p}{}{\p \imp \bot}$ to be the sole clause for $\p$ in a
definition; since there is a derivation for any sequent in which
$\bot$ appears in the assumption set, it follows easily from this that
a logic that permits such a definition is inconsistent.
The crux of the argument is to show that the sequent
$\seq{\cdot}{\p}{\bot}$ would have a derivation in this logic; from
this it follows easily that $\seq{}{}{\p \imp \bot}$ and,
therefore, $\seq{}{}{\p}$ have derivations and then, using
the multicut rule, we may derive $\seq{}{}{\bot}$.
To see that $\seq{}{\p}{\bot}$ has a derivation, we observe that
it would have one if $\seq{}{\p,\p}{\bot}$ is derivable.
A derivation of the last sequent is easily obtained by using the
clause defining $\p$ with one of the two occurrences of $\p$ in the
assumption set of the sequent.

The ground stratification condition that we have described disallows
clauses of the kind just discussed in definitions.
For such a clause to be allowable, we would need to be able to assign
a level to $\p$ that would have the property of being greater than or
equal to the level of $\p \imp \bot$.
However, the definition of levels for formulas makes this impossible
to do. 

In the example considered, ground stratification boils down to a form
of stratification that uses only the predicate name; this is because
the clause in question is for a predicate with no arguments.
When we consider predicates with arguments, the two notions become
distinct. 
In such a case, we may ``build'' the arguments of the predicate into
its name.
Of course, this is only possible for ground instances of predicates,
but that is sufficient for the coherence of the idea: specifically, we
can identify the well-formedness of a definitional clause 
with the ability to describe a stratification ordering on all its
instances. 
This is the idea that is reflected in ground stratification.
As a concrete example, consider a definition that is given by the
following clauses:
\begin{center}
  \begin{tabular}{ccc}
    $\fpdfn{\ev \app \zero}{}{\top}$ & \qquad\qquad
      $\fpdfn{\ev\app (\suc\app X)}{X}{\ev\app X \supset \bot}$
  \end{tabular}
\end{center}
In these clauses, \zero\ and \suc\ are to be understood to be
constants of a designated type \nat\ and $X$ should be read as
a variable of type \nat; intuitively, \ev, which is a
predicate of type $\nat \rightarrow \o$, identifies
the even natural numbers, which are represented by terms
constructed using the constant $\z$ (representing the numeral $0$)
and the function symbol $\suc$ (representing the successor function).
Now, if we were to consider an ordering of ground atomic formulas of
the form $\ev\app t$ that ``forgets'' the argument, the
displayed clauses would not satisfy the stratification requirement.
However, if we allow the chosen ordering to also take into account the
complexity of the argument, as would be the case if the measure
associated with $\ev\app t$ is based solely on the
complexity of $t$ as a term, then it is easily seen that the
definition will satisfy the required condition.

In the above example, the clauses for $\ev$ are assumed to provide
a fixed-point definition for the predicate.
The requirements we have described in Section~\ref{sec:logic} do not
allow them to be treated inductively: for such an interpretation, the
relevant definitional clauses must satisfy the stronger condition of
being stratified based on an ordering that uses only the predicate
name for ground formulas of the form $\ev\app t$.
A natural question to ask is if this requirement on inductive predicates
can be weakened to allow the ordering to depend on the arguments.
What we observe below is that such a weakening is not possible because
it can render the logic inconsistent. 

If we were to weaken the requirement in the way described, it would
allow for the definition of the predicate $\odd$ through 
the following clause:
\[\inddfn{\odd\app (\suc\app X)}{{X}}{\odd\app X\imp \bot}\]
Recast into a form based on a fixed-point operator, this definition
would correspond to one of the form
$\inddfn{\odd\app X}{X}{((B\app \odd)\app X)}$
where $B$ is the term
$\lambdax{p}
          {\lambdax{x}
                   {\existsx{y}
                            {(\eq\app x\app (\suc\ y)
                              \land (p\app y\imp \bot))}}}.$
A crucial observation about this operator is that does not impose the
requirement that the predicate it applies to also holds of $\suc\app \zero$,
\ie, of the representation of the numeral $1$.
Given the definition of \ev, it is therefore not surprising
that both
$\lambdax{x}{\ev\app x}$ and
$\lambdax{x}{(\ev\app x \imp \bot)}$
describe fixed points by virtue of it.
Concretely, it is easily seen that if we pick $S$ to be either of
these predicates then the sequent $\seq{x}{B \app S\app x}{S\app x}$ is
derivable.
Using the induction rule, it then follows that, for any term $t$, we
can construct derivations for
$\seq{{\mc X}}{\odd\app t}{\ev\app t}$ and 
$\seq{{\mc X}}{\odd\app t}{\ev\app t \imp \bot}$ and
therefore, by the multicut rule, a derivation for
$\seq{{\mc X}}{\odd\app t}{\bot}$.
As particular instances of this observation, we see that the sequents
$\seq{}{\odd\app \zero}{\bot}$ and
$\seq{}{\odd\app (\suc\app \zero)}{\bot}$ have derivations. 
The former implies that $\seq{}{}{\odd\app \z \imp \bot}$
has a derivation and therefore, by the clause defining \odd, that
$\seq{}{}{\odd\app (\suc\app \zero)}$ has a
derivation.
Since both $\seq{}{}{\odd\app (\suc\app \zero)}$
and 
$\seq{\cdot}{\odd\app (\suc\app \zero)}{\bot}$ have
derivations, it follows by the multicut rule that
$\seq{\cdot}{\cdot}{\bot}$ must also be derivable.

In the positive direction, the results in this paper show that the
restrictions we have described suffice to yield a consistent logic.
These results strengthen those in \cite{Tiu12} by adding a
treatment of inductive definitions to the logic.
A question to ask is if the conditions on inductive clauses are not
too restrictive in practice.
We argue not.
In particular, we believe that the weaker form of stratification suffice
for general fixed-point definitions and that inductive
reasoning can generally be based on definitions that satisfy stronger
conditions.
In the present context, for example, inductive reasoning can be based
on a definition of a predicate that characterizes natural
numbers---which is easily seen to be stratifiable under an ordering
for atomic formulas that uses only the predicate name---rather than the
definitions of $\ev$ or $\odd$.
\section{Applications for the More Permissive Form for Definitions}\label{sec:examples}

The main difference between the definitions permitted in \LDI and
in \G is in the use of ground stratification as opposed to strict
stratification to determine well-formedness.
Several examples have been presented in~\cite{Tiu12} to illustrate the
benefits of the added richness.
These examples have included the encoding of the G\"odel-Gentzen
translation and the Kolmogorov double negation translation of
classical logic in intuitionistic logic.
Another important class of applications concerns the encoding of
logical relations, which are commonly used to reason about properties
of programs and programming languages.
For example, using the more permissible form of definitions in \LDI,
we can encode logical equivalence of programs that asserts that two
functions are equivalent if their behavior on equivalent arguments is
the same~\cite{Ha16}.
A more elaborate use of logical relations in this kind of reasoning
may be found in \cite{WN16}, where a step-indexed logical equivalence
relation is used to prove correctness of compiler transformations for
functional programs.
Although a different kind of extension to \G, one based on adding
rewriting and using that to encode definitions~\cite{BaNa12}, was used
to validate this application, the use of ground stratification
provides an alternative justification that retains the flavor of
definitions as originally described by McDowell and Miller.

Another common use of logical relations is in the proof of strong
normalizability results in the style of Tait~\cite{Ta67}.
These arguments are based on the identification of reducibility
relations~\cite{GTL89}, whose encoding in a logic of definitions
requires the idea of ground stratification.
We will illustrate this application in more detail below.
We remark that logical relations arguments are
frequently used in conjunction with inductive arguments.
We will therefore observe that the pitfall discussed in the previous
section can be avoided, since the logical relations in question need
not be defined inductively. 

We will specifically show how we can encode the reducibility predicate
commonly used to prove strong normalization for the simply typed
lambda calculus (STLC) using ground stratification.
To do this, we introduce a base type $\ty$ to encode types in STLC,
which may be constructed from the following constant symbols
\begin{mathpar}
  \unit : \ty\and \arr : \ty\to\ty\to\ty
\end{mathpar}
as well as a type $\tm$ to encode terms, which may be constructed
from the following constant symbols
\begin{mathpar}
  \star : \tm\and \lamtm : (\tm\to\tm)\to\tm\and \apptm : \tm\to\tm\to\tm
\end{mathpar}
where $\star$ is the unique term of type $\unit$. We then assume that
we have already defined a predicate $\step : \tm\to\tm\to\o$ specifying
a reduction relation on terms, which we use to inductively define the
strong normalizability predicate $\sn : \tm\to\o$ as follows
$$\sn\ T\ind \typedforallx{\tm}{u}{(\step\ T\ u\imp \sn\ u)}$$
Using this, we can now define the reducibility predicate
$\red : \ty\to\tm\to\o$ as follows
\begin{align*}
  \red\ \unit\ T &\defn \sn\ T\\
  \red\ (\arr\ A\ B)\ T &\defn \typedforallx{\tm}{u}{
                          (\red\ A\ u\imp \red\ B\ (\apptm\ T\ u))}
\end{align*}
Note that this predicate cannot be strictly stratified because $\red$
appears negatively in the body of its definition.
However, we notice that in the instance $\red\ A\ u$ appearing
negatively in the body, the type argument $A$ is a subterm of the type
argument $(\arr\ A\ B)$  appearing in the head of the definition.
But if we think of $\red$ 
as being a definition indexed by its first type argument, then we
can think of $\red\ A$ as having already been defined when $\red\ (\arr\ A\ B)$
is being defined, suggesting the definition be acceptable.
Indeed, this definition is ground stratified because for any ground
terms $a$ and $b$ of type $\ty$, the term $a$ is strictly smaller than
the term $(\arr\ a\ b)$, and thus our definition is valid. 

One may wonder whether this definition suffices for the usual reducibility
argument to go through. Indeed, it had been assumed in a
previous development that it was safe to do induction directly on
the reducibility predicate \cite{AbellaWeb}, which violates the
strict stratification condition we have now imposed on inductive definitions.
The usual argument (see eg. \cite{GTL89}), however, depends on induction on
the {\em type} argument, and not on reducibility itself. We can therefore
model this mode of reasoning by defining a predicate $\type : \ty\to\o$
inductively as follows
\begin{align*}
  \type\ \unit &\ind \top\\
  \type\ (\arr\ A\ B) &\ind (\type\ A)\wedge(\type\ B)
\end{align*}
and then proceed by induction on the type, rather than on
reducibility.
We note that the \type\ predicate does not appear negatively in the
body of these clauses and hence these clauses satisfy the
stratification condition under a level assignment in which the level
of an atomic predicate of the form $(\type\ t)$ is determined solely by
the constant \type. 

\section{Proving Consistency for the Logic}\label{sec:consistency}

In this section, we sketch the proof of consistency for the logic
\LDI.
The most common method for doing this is to show that there cannot be
a derivation for $\bot$ in the logic.
In a sequent style formulation, it is usually easy to show this if
proofs are limited to those that do not use the cut rule.
Thus, a cut-elimination result for the logic comes in handy. 
However, as with the logic \LD described by Tiu, a proof of
cut-elimination for \LDI has been elusive.
Following Tiu, we therefore use an indirect technique.
We first describe a ground logic \LDIinf through a sequent calculus
that follows the structure of the one for 
\LDI except that it is specialized to proving ground sequents.
We then show that cut-elimination holds for \LDIinf and use this to
conclude that there is no derivation for $\bot$ in this logic. 
Finally, we produce an interpretation of \LDI into \LDIinf, which
allows us to conclude that there cannot be a derivation for $\bot$ in
\LDI either.
The subsections below elaborate on these three steps in the proof.

\subsection{The Ground Logic}
\begin{figure}[htb!]
  \begin{mathpar}
    \inferrule*[right=\allR]
    {\left\{\gseq{\Gamma}{C[t/x]_\emptyset}
      \right\}_{t\in\ground(\alpha)}}
    {\gseq{\Gamma}{\typedforallx{\alpha}{x}{C}}}\and
    \inferrule*[right=\exL]
    {\left\{\gseq{\Gamma, B[t/x]_\emptyset}{C}
      \right\}_{t\in\ground(\alpha)}}
    {\gseq{\Gamma, \typedexistsx{\alpha}{x}{B}}{ C}}
  \end{mathpar}
  \caption{\allR and \exL rules in \LDIinf}
  \label{fig:omega-rules}
\end{figure}

We now define a ground version of \LDI which we call \LDIinf. The formulas in
\LDIinf are the formulas in \LDI which are well formed over the empty variable
context. Since all formulas are ground, we restrict the sequent $\seq{\mc X}{\Gamma}{ C}$
to have no eigenvariables, which translates to requiring that $\mc X = \emptyset$.
We therefore denote a ground sequent by $\gseq{\Gamma}{ C}$. The rules of \LDIinf are
the same as those of \LDI except for \allR, \exL, and \muL. We replace the \allR
and \exL rules with the ground infinitary rules in Figure~\ref{fig:omega-rules}.

\begin{figure}[htb!]
  \begin{mathpar}
  {\inferrule*[right=\DL]
    {\{\gseq{\Gamma,B'}{ C}\st \dfn(H\defn_{\mc X} B,A,\epsilon_\emptyset,B')\}}
    {\gseq{\Gamma, A}{ C}}}\and
  {\inferrule*[right=\DR]
    {\gseq{\Gamma}{ B'}}
    {\gseq{\Gamma}{ A}}}\quad\raisebox{1em}{$\dfn(H\defn_{\mc X} B, A, \epsilon_\emptyset,B')$}
\end{mathpar}
\caption{Definition rules in \LDIinf for a predicate $p$, provided $A = p\ \vec t$}
\label{fig:ground-def-rules}
\end{figure}

The definition \DL and \DR rules in \LDIinf are the same as those of \LDI. The
restriction to ground sequents allows us to present these rules as shown in
Figure~\ref{fig:ground-def-rules}. However, we must check that these rules
are well-formed, which amounts to checking that the premises are well-formed
over the empty variable context. In other words, assuming that the conclusion of
the rule is ground, we must check that all the formulas in each of the premises
are ground. Suppose we are given a ground context $\Gamma$, a ground formula $A$,
and a clause $H\defn_{\mc X} B$ such that $\dfn(H\defn_{\mc X}B, A,\epsilon_\emptyset,B')$,
then there exists $\rho : \mc Z\to \mc X$ such that $A = A\epsilon = H\rho$ and
$B\rho = B'$. Since $H$ lies in the higher-order pattern fragment and $A$ is ground,
this implies that $\rho$ is unique and furthermore that it must be ground.
Since $B$ lies over $\mc X$, which is the domain of $\rho$, it follows that $B'$
must be ground, as desired.

\begin{figure}[htb!]
  \begin{mathpar}
    \inferrule*[right=\IL]
    {\left\{{B\ S\ \vec t\vdash S\ \vec t}
      \right\}_{\vec t\in\ground(\vec \alpha)}\\
      {\Gamma, S\ \vec u\vdash C}}
    {\Gamma, p\ \vec u\vdash C}
\end{mathpar}
\caption{\IL rule in \LDIinf for $p\ \vec x\ind_{\vec x} B\ p\ \vec x$ and inductive invariant $S$}
\label{fig:ind-inf-rule}
\end{figure}

Finally, for any inductive definition with fixed-point form $p\ \vec x\ind_{\vec x} B\ p\ \vec x$
and inductive invariant $S$, we introduce a ground infinitary version of the \IL
rule given in Figure~\ref{fig:ind-inf-rule}.

\subsection{Consistency of the Ground Logic}

We now show that \LDIinf is consistent by proving the cut-elimination theorem
for ground derivations. To do so, we first define a {\em reduction} relation,
which specifies how the multicut rule may be propagated upwards in a derivation
tree. Then, we introduce the notion of normalizable derivations, which allows
us to show that the cut reduction relation is well-founded, and from which we
can establish the cut-elimination theorem. To prove that every derivation is
normalizable, we require the intermediate notion of a reducible derivation.
Every derivation is then shown to be reducible via the reducibility technique,
inspired by Tait and Martin-L\"of (see \cite{Ta67,ML71}), and adapted to the
current context by McDowell and Miller in \cite{MDM00}. Our proof closely follows
that of \cite{Tiu12} and \cite{MDM00}, together with the treatment of induction
from \cite{Tiu04,MT04}. We outline the essential aspects of these
steps below.

\subsubsection{Cut reductions}
We specify the reduction relation between derivations,
following closely the reduction relation in \cite{MDM00}. The \emph{redex},
that is the derivation to be reduced, is always a derivation $\Xi$ ending with
the multicut rule
\begin{mathpar}
{\inferrule*[right=\multc]
        {\deriv{\Pi_1}{\gseq{\Delta_1}{B_1}}
        \\ \cdots
        \\ \deriv{\Pi_n}{\gseq{\Delta_n}{ B_n}}
        \\ \deriv{\Pi}{\gseq{B_1,\ldots,B_n,\Gamma}{ C}}
        }
        {\gseq{\Delta_1,\ldots,\Delta_n,\Gamma}{ C}}}
\end{mathpar}
We refer to the formulas $B_1,\dots,B_n$ in the multicut as {\em cut formulas}.
If a left or right rule introduces a cut formula, we say the rule is \emph{principal}.
If $\Pi$ reduces to $\Pi'$, we say that $\Pi'$ is a \emph{reduct} of $\Pi$.

The reduction relation relates a derivation $\Xi$ ending with the multicut
rule to a new derivation $\Xi'$ which may again end with a multicut but is
always of smaller complexity. The reduction relation is specified by case analysis
on the last rule of $\Pi$. If $\Pi$ is principal, then we group the reductions
by case analysis on the derivation $\Pi_i$ ending with the cut formula.
These are {\em essential cases} whenever $\Pi_i$ is also principal (but not \muR),
{\em left-commutative cases} whenever $\Pi_i$ is not principal,
{\em inductive cases} whenever $\Pi$ ends with \muL, {\em structural cases}
whenever $\Pi$ ends with a weakening or contraction, {\em left axiom cases}
whenever $\Pi_i$ is an axiom, and {\em left multicut case} whenever $\Pi_i$ ends
with a multicut. Otherwise, $\Pi$ is either an axiom ({\em right axiom case}),
ends with a multicut ({\em right multicut case}), or ends with a non-principal rule
({\em right-commutative cases}). We describe the inductive case in
more detail below. The full reduction relation is presented in a
longer version of this paper.  

  \case{\em \underline{Inductive case}}
  An inductive case occurs when $\Pi$ ends with \muL introducing the
  cut formula $p\ \vec t$ for an inductive definition in fixed-point form
  $p\ \vec x \ind_{\vec x} B\ p\ \vec x$. Suppose the cut formula being
  introduced is $B_1 = p\ \vec t$. In this case, $\Pi_1$ is a derivation of
  the sequent $\gseq{\Delta_1}{ p\ \vec t}$ and $\Pi$ is the following derivation
  for some inductive invariant $S$
\begin{mathpar}
  {\inferrule*[right=\IL]
    {\left\{\deriv{\Pi^{\vec u}_S}
        {\gseq{B\ S\ \vec u}{ S\ \vec u}}\right\}_{\vec u\in\ground(\vec \alpha)}\\
      \deriv{\Pi'}{\gseq{S\ \vec t,B_2,\ldots,B_n,\Gamma}{ C}}}
    {\gseq{p\ \vec t,B_2,\ldots,B_n,\Gamma}{C}}}
\end{mathpar}
The key idea is that since we know that $S$ is an inductive invariant, or pre-fixed-point of $B$, and $p$ is a least fixed-point of $B$,
then any time $p\ \vec t$ is provable in some context,  $S\ \vec t$
should also be provable in the same context.
More specifically, the family of derivations $\{\Pi^{\vec u}_S\}_{\vec u}$,
which we abbreviate by $\Pi_S$, may be used to obtain a derivation of
$\Delta_1\vdash S\ \vec t$. We call this derivation the \emph{unfolding}
of the derivation $\Pi_1$ with respect to the family $\Pi_S$, which is captured in
the following lemma, and which we denote by $\mu(\Pi_1,\Pi_S)$.

\begin{lemma}[Unfolding lemma]
  Suppose $p\ \vec x\ind B\ p\ \vec x$ is the fixed-point form of
  an inductive definition, then
  for any derivation $\Psi$ of $\Delta\vdash D\ p$ where $p$ does not
  occur in $D$ and occurs only positively in $D\ p$ (\ie, does not occur to
  the left of an implication), there exists a derivation $\mu(\Psi,\Pi_S)$
  of $\Delta \vdash D\ S$.
\end{lemma}

The proof of this lemma depends crucially on the inductive definition $p$
to be strictly stratified. Using the unfolding operation, we can now
reduce the redex $\Xi$ above to the following 
\begin{mathpar}
  {\inferrule*[right=\multc]
    {{\deriv{\mu(\Pi_1,\Pi_S)}{\gseq{\Delta_1}{S\ \vec t}}}\\
      \ldots\\ 
      {\deriv{\Pi'}{\gseq{S\ \vec t,B_2,\ldots,B_n,\Gamma}{ C}}}}
    {\gseq{\Delta_1,\ldots,\Delta_n, \Gamma}{C}}}
\end{mathpar}

\subsubsection{Cut Elimination}

Our aim is now to show how the cut reduction relation is well-founded, and can therefore be used to eliminate cuts from a derivation. 
We start with the following inductive definition, by which we
mean the smallest set of derivations closed under the specified operation
(see definition 1.1.1 in \cite{A77} for a precise definition).

\begin{definition}[normalizability]
  The set of normalizable derivations in \LDIinf is inductively defined as follows:
  \begin{enumerate}
  \item if $\Pi$ ends with a multicut, then $\Pi$ is normalizable if
    every reduct $\Pi'$ is normalizable
  \item otherwise, $\Pi$ is normalizable if each of its premise derivations is normalizable.
  \end{enumerate}
\end{definition}

The goal is to show that every derivation in \LDIinf is normalizable.
To do this we need the intermediate notion of $\gamma$-reducibility,
which is defined by transfinite recursion on the level $\gamma$ of a derivation,
which we recall is an ordinal. The \emph{level} of a derivation of a ground sequent
$\gseq{\Gamma}{ C}$ is defined to be $\lvl(C)$. Note therefore that if $\Pi$
reduces to $\Pi'$, then $\lvl(\Pi) = \lvl(\Pi')$. Furthermore, note that
levels are defined so that every rule has a non-increasing consequent in each
of its major premises. 
In the case of the \DR and \muR rules, this condition is
guaranteed by the ground stratification condition. In the case of \DL,
this condition is guaranteed because the consequent is ground.
Finally we note that below,
$\gamma$-reducibility of a derivation ending with the \impR rule depends
on $\alpha$-reducibility to already be defined for any $\alpha < \gamma$.
Together, these observations guarantee that the following is well defined.

\begin{definition}[reducibility]
  Define $\gamma$-reducibility of a derivation $\Pi$ inductively as follows:
  \begin{enumerate}
  \item if $\Pi$ ends with a multicut, then $\Pi$ is $\gamma$-reducible if
    for every reduct $\Pi'$, $\Pi'$ is $\gamma$-reducible. 
  \item if $\Pi$ is the derivation
    $$\inferrule*[right=\impR]
    {{\deriv{\Pi'}{\gseq{\Gamma,A}{ B}}}}
    {\gseq{\Gamma}{ A\imp B}},$$
    then $\Pi$ is $\gamma$-reducible if $\lvl(\Pi)\leq \gamma$,
    $\Pi'$ is $\gamma$-reducible, and for every $\alpha$-reducible
    $\deriv{\Upsilon}{\Delta\vdash A}$ where $\alpha = \lvl(A)$,
    the derivation $\multcut(\Upsilon,\Pi')$
    is $\gamma$-reducible.
  \item if $\Pi$ ends with any other rule, $\Pi$ is $\gamma$-reducible if each
    of its major premise derivations is $\gamma$-reducible, and each of its
    minor premise derivations is normalizable.
  \end{enumerate}
\end{definition}

We say a derivation $\Pi$ is \emph{reducible} if there exists a $\gamma$
such that $\Pi$ is $\gamma$-reducible. The lemma below asserts that
reducibility implies normalizability. Since reducibility is a strengthening
of normalizability, it is shown by straightforward induction on the
$\gamma$-reducibility of a derivation.
\begin{lemma}[Normalization Lemma]
  \label{normalization-lemma}
  If $\Xi$ is $\gamma$-reducible then $\Xi$ is normalizable.
\end{lemma}

The key lemma needed to prove normalizability of derivations in \LDIinf is the following
\begin{lemma}[Reducibility Lemma]
  \label{reducibility-lemma}
  Given $n$ derivations $\Pi_1,\ldots,\Pi_n$ such that $\Pi_i$ is $\gamma_i$-reducible,
  and any derivation $\Pi$ the multicut
  \begin{mathpar}
    {\inferrule*[right=\multc]
      {\deriv{\Pi_1}{\gseq{\Delta_1}{ B_1}}\\\ldots\\
        \deriv{\Pi_n}{\gseq{\Delta_n}{ B_n}}\\
      \deriv{\Pi}{\gseq{\Gamma,B_1,\ldots,B_n}{ C}}}
      {\gseq{\Delta_1,\ldots,\Delta_n,\Gamma}{ C}}
    }
  \end{mathpar}
  which we denote by $\Xi$, is reducible.
\end{lemma}
Following \cite{Tiu04}, the proof is by induction on (1) the number of
occurrences of \IL in $\Pi$, then by subordinate induction on (2) the height of $\Pi$, then on (3) the number of cut formulas $n$, and finally on (4) the $\gamma_i$-reducibility of each derivation $\Pi_i$.

We now proceed to show how the multicut rule may be eliminated from any \LDIinf derivation.
\begin{corollary}
  \label{reducibility-corollary}
  Every derivation $\Pi$ in \LDIinf is reducible.
\end{corollary}
\begin{proof}
  Consider the nullary multicut $\multcut(\Pi)$, which is reducible
  by Lemma~\ref{reducibility-lemma}. Since $\multcut(\Pi)$ reduces to $\Pi$
  it follows that $\Pi$ is reducible.
\end{proof}

\begin{lemma}[Normal form lemma]
  \label{normal-form-lemma}
  If a derivation $\Pi$ of a sequent $\Gamma\vdash C$ is normalizable,
  then there exists a cut-free derivation $\hat\Pi$ of $\gseq{\Gamma}{ C}$,
  which we call a \emph{normal form} for $\Pi$.
\end{lemma}
Since every derivation in \LDIinf is reducible by Corollary~\ref{reducibility-corollary},
reducibility implies normalizability by Lemma~\ref{normalization-lemma},
and a normalizable derivation has a cut-free normal form by the previous lemma,
we obtain our main theorem
\begin{theorem}[Cut admissibility for \LDIinf]
  \label{main-theorem}
  Every derivation of a ground sequent $\gseq{\Gamma}{ C}$ in \LDIinf admits
  a cut-free derivation of the same sequent.
\end{theorem}
We now obtain the following important consequence of the cut elimination theorem.
\begin{corollary}[Consistency of \LDIinf]
  \label{ground-consistency}
  \LDIinf is consistent.
\end{corollary}
\begin{proof}
  It suffices to notice by case analysis that there are no cut-free
  derivations of $\vdash \bot$. Since every derivation has a cut-free normal
  form by Theorem~\ref{main-theorem}, $\vdash\bot$ is not derivable in \LDIinf.
\end{proof}

\subsection{Consistency of the Full Logic}

The goal is now to define an interpretation of \LDI in \LDIinf,
which will relate the derivability of a sequent in \LDI to the derivability
of a ground sequent in \LDIinf. This will allow us to reduce the consistency
of the former to that of the latter. The interpretation is specified
by the following lemma.
\begin{lemma}[Grounding Lemma]
  If $\seq{\mc Y}{\Gamma}{ C}$ is derivable in \LDI, then for any $\mc Y$-grounding
  substitution $\delta : \emptyset\to\mc Y$, the ground sequent
  $\gseq{\Gamma\delta}{C\delta}$ is derivable in \LDIinf. 
\end{lemma}
\begin{proof}
  The proof is by induction on the height of the derivation of $\seq{\mc Y}{\Gamma}{ C}$, and by case analysis on the last rule of the derivation.
  We only show the case where the derivation ends with \DL below.
  \case{\underline{\em Case \DL}} Suppose the derivation ends with \DL on the atom $A$
  where $\Gamma = \Gamma', A$. It is clear that if
  $\dfn(H\defn_{\mc X} B, A\delta, \epsilon_\emptyset, B')$,
  then $\dfn(H\defn_{\mc X} B, A, \delta, B')$,
  and furthermore that $\range(\delta) = \emptyset$ since $\delta$ is a
  ground substitution. Thus a ground derivation of a premise
  $\gseq{\Gamma'\delta, B'}{ C\delta}$ below is obtained by applying
  the induction hypothesis to the ground premise
  $\seq{}{\Gamma'\delta,B'}{ C\delta}$ with the empty grounding substitution. 
  \begin{mathpar}
    {\inferrule*[right=\DL]
      {\{\seq{\range(\theta)}{ \Gamma'\theta, B'}{ C\theta}\st
        \dfn(H\defn_{\mc X} B, A, \theta, B')\}}
      {\seq{\mc Y}{ \Gamma', A }{ C}}
    }\quad\leadsto\quad
    {\inferrule*[right=\DL]
      {\{\gseq{\Gamma'\delta, B' }{ C\delta}\st
        \dfn(H\defn_{\mc X} B, A\delta, \epsilon_\emptyset,B')\}}
      {\gseq{\Gamma'\delta, A\delta }{ C\delta}}
    }
  \end{mathpar}
  The effect of the interpretation on a derivation is therefore to
  prune branches from the \DL rule.
\end{proof}
As a result of the grounding lemma, we obtain the following
\begin{lemma}[Interpretation lemma]
  \label{interpretation}
  If \LDIinf is consistent, then so is \LDI.
\end{lemma}
\begin{proof}
  From grounding lemma, if $\vdash \bot$ is derivable in \LDI, then it is
  also derivable in \LDIinf. Thus, the contrapositive also holds:
  If $\vdash \bot$ is not derivable in \LDIinf, then it is not derivable in \LDI.
\end{proof}
The interpretation lemma in conjunction with the consistency of
\LDIinf (Corollary~\ref{ground-consistency}) allow us to obtain
\begin{corollary}[Consistency of \LDI]
  \LDI is consistent.
\end{corollary}

\section{Conclusion}\label{sec:conclusion}

In this paper, we have shown how to accommodate inductive definitions
in a logic that permits a more permissive form of fixed-point
definitions.
We have observed that the provision of inductive definitions must be
done with some care in order to ensure consistency.
We have also seen how the resulting logic allows us to encode
arguments based on logical relations that are often used in
formulating and proving properties about programs and programming
languages.
This work represents an intermediate step towards building a more
flexible form for definitions into the logic underlying the
Abella proof assistant.
That logic additionally includes co-induction, generic quantification,
and a notion called \emph{nominal abstraction} which provides a means
for analyzing occurrences of objects in expressions that are governed
by generic quantifiers. 
In a longer version of this paper, we have shown that the greater
flexibility in fixed-point definitions can be supported even in the
presence of generic quantification.
The extension of these results to include the remaining  features of
the logic is the subject of ongoing work.

The work that we have described builds on that of Tiu.
Another effort that has a related flavor is that of Baelde and
Nadathur~\cite{BaNa12}.
That effort resulted in the development of a natural deduction
calculus called \NJI, which extends deduction
modulo~\cite{dowek03jsl} with inductive and co-inductive definitions. 
This calculus allows fixed-point definitions to be constructed with a
flexibility that overlaps significantly with what is supported by
ground stratification, with the difference that such definitions are
realized through rewrite rules.
An important aspect of \NJI is an equality elimination rule that
builds in the ability to generalize equality assumptions, which then
allows substitutions into proofs to be defined in a way that maintains
their original structure.
This feature has been exploited in showing strong normalizability for
\NJI, thereby verifying its consistency.
While \NJI includes a treatment of co-induction, it does not support
generic quantification and nominal abstraction.
We believe that the approach underlying this paper may be a preferred
way to provide for the richer form of definitions because it does not
permit rewrite rules that can modify the meaning of equality.
However, there are insights to be gained from the formulation of
equality elimination in \NJI in developing a calculus close to \LDI
that also includes the additional features of interest and for when we
can prove a cut elimination result.

We would like to extend this work in the future in a few ways in
addition to including the features mentioned in the logic.
First, we would like to build greater flexibility into the treatment
of induction.
Currently, inductive clauses do not allow the predicate being defined
to appear to the left of an implication in the body.
However, we believe that this requirement can be weakened to prohibit
such occurrences only in truly negative positions.
Second, perhaps taking the cue from \NJI, we would like to develop a
non-ground calculus close to \LDI for which we can prove a
cut-elimination result.
Finally, in a much more ambitious direction, we would like to add a 
higher order quantification capability to logics of definition.
This would on the one hand allow for greater modularity in
definitions, statements of theorems and their proofs, and may, on the 
other hand, allow for the statement and proofs of stronger results
such as strong normalizability for System F~\cite{GTL89}.


\end{document}